\documentclass[a4paper]{article}
\usepackage{amsfonts,amssymb,amsmath,amsthm,cite}
\usepackage{ucs} 
\usepackage[utf8x]{inputenc}
\usepackage{graphicx}
\usepackage[english]{babel}
\usepackage{slashed}
\usepackage{textcomp}
\usepackage{bm}
\usepackage{titlesec}
\usepackage{tikz-cd}
\usepackage{adjustbox}
\usepackage{multirow}
\usepackage{etoolbox}
\patchcmd{\thebibliography}{\section*}{\section}{}{}
\usepackage{mathtools}
\mathtoolsset{showonlyrefs}
\evensidemargin=\oddsidemargin
\newtheorem{theorem}{Theorem}
\newtheorem{lemma}{Lemma}
\newtheorem{deff}{Definition}
\newtheorem{remark}{Remark}
\newtheorem{exm}{Example}

\begin{document}
\vspace{10mm}
\begin{center}
	\large{\textbf{An etude on a renormalization}}
\end{center}
\vspace{2mm}
\begin{center}
	\large{\textbf{Aleksandr V. Ivanov}}
\end{center}
\begin{center}
St. Petersburg Department of Steklov Mathematical Institute of Russian Academy of Sciences,\\ 
27 Fontanka, St. Petersburg 191023, Russia
\end{center}

\begin{center}
	E-mail: regul1@mail.ru
\end{center}
\begin{flushright}
On the 85th anniversary of PDMI RAS
\end{flushright}
\begin{flushright}
\adjincludegraphics[width = 6.3  cm, valign=c]{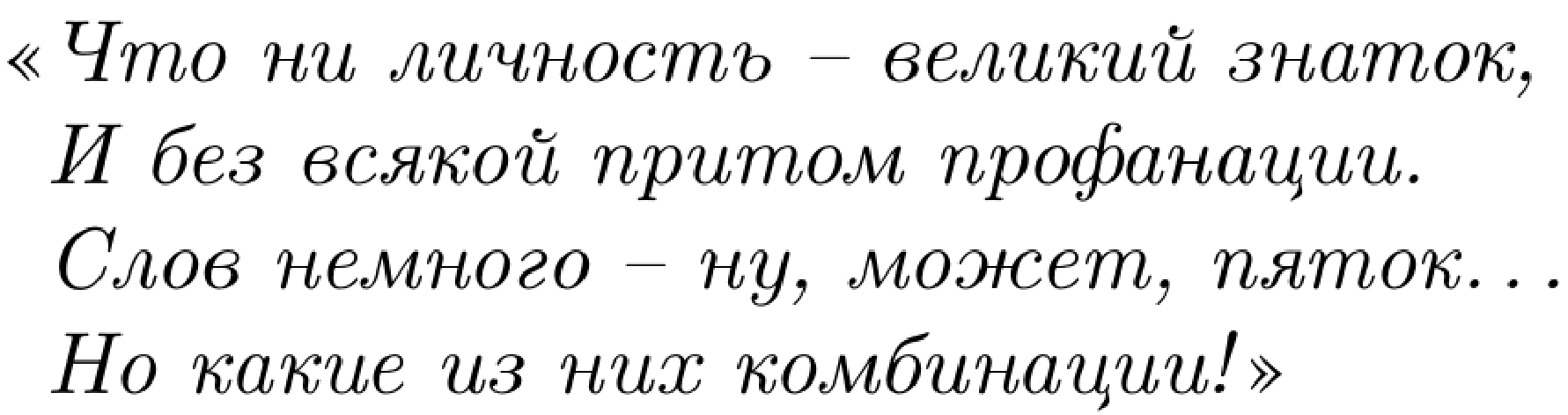}\\
\vspace{0.5mm}
A.A.Ivanov (1936--1996)
\end{flushright}

\vspace{10mm}

\textbf{Abstract.} In this paper, we study renormalization, that is, the procedure for eliminating singularities, for a special model using both combinatorial techniques in the framework of working with formal series, and using a limit transition in a standard multidimensional integral, taking into account the removal of the singular components. Special attention is paid to the comparative analysis of the two views on the problem. It is remarkably that the divergences, which have the same form in one approach, acquire a different nature in another approach and lead to interesting consequences. A special deformation of the spectrum is used as regularization.

\vspace{2mm}

\tableofcontents

\newpage	
	
\section{Introduction}
\label{34:sec:int}
Regularization and renormalization play an important role in modern theoretical physics and are used in a large number of quantum field models \cite{9,10,3,6,7,105}. However, in mathematical physics and pure mathematics, such concepts are rare or have a slightly different meaning, see \cite{Gelfand-1964,Vladimirov-2002}. First of all, this is due to the fact that the mathematical formalism devoted to a functional integration, see \cite{34-3,4,34-6}, is currently insufficiently developed. It is for this reason that the basic quantum field models have to be studied using the perturbative method, for which it is quite sufficient to understand the functional integral in a very stripped-down form.

From the point of view of the perturbative approach, which is associated with the analysis of formal (less often asymptotic) series, the concepts of regularization and renormalization are led by the presence of non-integrable densities and singular functionals in the coefficients of the series under study. This approach in this article will be called "physical". However, within the framework of another approach, which will be called "mathematical", in those models in which a strict definition of the functional integral can be used, formal series no longer appear, and consequently, the reasons for the introduction of the regularization and the application of the renormalization become different. Note that the words "physical" and "mathematical" are used only to distinguish, both approaches have a neat mathematical form.

In this paper, we analyze a special type of model that allows for both the physical description and a formulation using the functional integral in the sense of limit transition in a finite-dimensional case, see \cite{daniell-1919,Iv-19,zinn-2005}. The latter theme is quite popular and is used for various models \cite{34-c-m,I-R,34-11,34-12,34-13}. This formulation of the problem allows us to trace all the subtleties of both approaches, make a comparative analysis, and identify the positive and negative sides. Remarkably, the divergences that arise when using the perturbative approach due to the appearance of non-integrable densities take on a completely different meaning in the second approach: they manifest themselves either by the presence of a rapidly oscillating phase, or by the zeros of the entire functional. At the same time, the regularization and the renormalization are performed in the same way in both approaches.

The purpose of this paper is to demonstrate the application of mathematical formalism to a process that often occurs in theoretical physics, as well as an attempt to draw the attention of researchers working in related mathematical fields to an existing interesting problem. The mathematical analysis of the model presented in the article can be considered a brief outline, demonstrating in what form the renormalization procedure would like to generalize to a wider class of models.

The work has the following structure. Section \ref{34:sec:m} studies the case using a special type of functional integral. Necessary definitions and proofs are provided. Concepts such as the regularization, the renormalization, and the limit transition are discussed. Section \ref{34:sec:p} examines an approach based on the use of formal series. Concepts such as a diagrammatic technique, as well as the regularization and the renormalization are discussed. The section also contains a comparison of approaches. Section \ref{34:sec:zac} contains some concluding remarks.

\section{Mathematical approach}
\label{34:sec:m}

\begin{deff}\label{34-d-6}
Let $V$ denote a set of sequences $a=\{a_j\}_{j=1}^{+\infty}$, the elements $a_j$ of which belong to $\mathbb{R}$. Let also $n\in\mathbb{N}$. The symbol $p_n$ denotes a projector from $V$ to $V$ acting on an arbitrary $a\in V$ according to the rule:
\begin{equation}\label{34-21}
\big(p_n(a)\big)_j=a_j\,\,\,\mbox{for}\,\,\,j\in\{1,\ldots,n\},\,\,\,
\big(p_n(a)\big)_j=0\,\,\,\mbox{for}\,\,\,j>n.
\end{equation}
\end{deff}
\begin{deff}\label{34-d-7}
Let $k\in\mathbb{N}$. With the symbol $\mathcal{B}_k$, we notate a subset of $V$, all elements $\beta\in\mathcal{B}_k$ of which satisfy two conditions:
\begin{enumerate}
	\item $\beta_j>0$ for all $j\in\mathbb{N}$;
	\item $b_k(\beta)=\sum_{j=1}^{+\infty}\beta_j^{-k}<+\infty$.
\end{enumerate}
In addition, we use the symbol $\mu(\beta)$ to notate the minimum element, that is $\min_{j>0}(\beta_j)$.
\end{deff}
\begin{remark}\label{34-r-11}
If $\beta\in\mathcal{B}_k$ for some $k\in\mathbb{N}$, then $\beta\in\mathcal{B}_{j}$ for all $k\leqslant j\in\mathbb{N}$. Also note that, if necessary, the elements of the sequence $\beta$ can be ordered in ascending order, since the only point of accumulation is infinity.
\end{remark}
\begin{deff}\label{34-d-2}
Let $n,k\in\mathbb{N}$, $s\in\mathbb{R}$, and also $\beta\in\mathcal{B}_k$. Let us introduce the function
\begin{equation}\label{34-1}
\Phi_n(s,\beta)=
\prod_{j=1}^{n}\int_{\mathbb{R}}\frac{\mathrm{d}a_j}{\sqrt{\pi/\beta_j}}
\,e^{-\beta_j^{\phantom{1}}a_j^2+isa_j^2}.
\end{equation}
\end{deff}
\begin{remark}\label{34-r-1}
Note that the integral representation \eqref{34-1} is set correctly for all $s\in\mathbb{C}$, such that $\mathrm{Im}(s)>-\mu(\beta)$. In this paper, the assumption $\mathrm{Im}(s)=0$  is not a significant limitation, since it can be achieved by changing the normalization of the integral and redefining the elements $\beta_j\to\tilde{\beta}_j=\beta_j-\mathrm{Im}(s)$.
\end{remark}
\begin{exm}\label{34-e-7}
In mathematical physics, the set $\beta$ is usually called a spectrum and is directly related to differential operators. Consider the following example. Let $\mathbb{N}\ni n>1$ and $\mathbf{B}=\{x\in\mathbb{R}^n:|x|\leqslant1\}$ is the $n$-dimensional ball. In the domain $\mathbf{B}$, we can consider a spectral problem with zero Dirichlet boundary condition
\begin{equation*}
-\sum_{i=1}^n\partial_{x_i}^2u(x)=\lambda u(x),\,\,\,
u(x)\big|_{|x|=1}=0.
\end{equation*}
It is known, see paragraph 4 of chapter 5 in \cite{34-14}, that this problem has a countable spectrum $\beta=\{\beta_j\}_{j=1}^{+\infty}$ and the corresponding set of eigenfunctions $\{\phi_j(x)\in L^2(\mathbf{B})\}_{j=1}^{+\infty}$. In this case, the set $\beta$ has a single accumulation point at infinity, and all its elements are strictly positive. The Green's function for the mentioned operator in this case has the form
\begin{equation*}
G(x,y)=\sum_{j=1}^{+\infty}\phi_j(x)\beta^{-1}_j\phi_j(y),
\end{equation*}
where $x,y\in\mathbf{B}$ and $x\neq y$. Let us pay attention to the behavior of the Green's function near the diagonal $x\sim y$, see \cite{29,30-1-1 }, which in the main order, when decomposed in $x-y$, is proportional to the function
\begin{equation*}
\ln(|x-y|)\,\,\,\mbox{for}\,\,\,n=2\,\,\,\mbox{and}\,\,\,
|x-y|^{2-n}\,\,\,\mbox{for}\,\,\,n>2.
\end{equation*}
It is clear that in this case $b_1(\beta)=+\infty$ for all $n>1$, since the value of $G(x,y)$ diverges at $y\to x$ inside the domain $\mathbf{B}$. The convergence of the remaining numbers $b_k(\beta)$, where $k>1$, is conveniently analyzed using the representation
\begin{equation*}
b_k(\beta)=\int_{\mathbf{B}}\mathrm{d}^nx_1\ldots\int_{\mathbf{B}}\mathrm{d}^nx_k\,
G(x_1,x_2)\cdot\ldots\cdot G(x_{k-1},x_k).
\end{equation*}
In particular, if $n=2,3,$ then in this case $\beta\in\mathcal{B}_2$ occurs due to the convergence of the last integrals.
\end{exm}
\begin{lemma}\label{34-l-3}
Let the assumptions of Definition \ref{34-d-2} be true. Then for function \eqref{34-1} the following representation is valid
\begin{equation}\label{34-3}
	\Phi_n(s,\beta)=
	\Bigg(\prod_{j=1}^{+\infty}\int_{\mathbb{R}}\frac{\mathrm{d}a_j}{\sqrt{\pi/\beta_j}}\Bigg)
	e^{-S_0[a,\beta]+isS_1[p_n(a)]},
\end{equation}
where auxiliary objects are defined by the equalities
\begin{equation}\label{34-2}
	S_0[a,\beta]=\sum_{j=1}^{+\infty}\beta_j^{\phantom{1}}a_j^2,\,\,\,
	S_1[a]=\sum_{j=1}^{+\infty}a_j^2.
\end{equation}
\end{lemma}
\begin{proof} Note that the object from \eqref{34-3} can be factored into two parts
\begin{equation*}
\Bigg(\prod_{j=1}^{n}\int_{\mathbb{R}}\frac{\mathrm{d}a_j}{\sqrt{\pi/\beta_j}}\,
e^{-S_0[a,\beta]+isS_1[p_n(a)]}\Bigg)
\Bigg(\prod_{j=n+1}^{+\infty}\int_{\mathbb{R}}\frac{\mathrm{d}a_j}{\sqrt{\pi/\beta_j}}\,
e^{-S_0[a-p_n(a),\beta]}\Bigg),
\end{equation*}
where the first part is the same as \eqref{34-1}, while the second part is equal to unit.
\end{proof}
\begin{lemma}\label{34-l-1} Let the assumptions of Definition \ref{34-d-2} be true. Then for function \eqref{34-1} the following representation is valid
\begin{equation}\label{34-5}
\Phi_n(s,\beta)=\prod_{j=1}^{n}\Big(1-is/\beta_j\Big)^{-\frac{1}{2}}=
f_n(s)e^{ig_n(s)/2},
\end{equation}
where
\begin{equation}\label{34-6}
f_n(s)=\prod_{j=1}^{n}\Big(1+s^2/\beta_j^2\Big)^{-\frac{1}{4}},\,\,\,
g_n(s)=\sum_{j=1}^n\arctan\big(s/\beta_j\big).
\end{equation}
\end{lemma}
\begin{proof} The first equality in \eqref{34-5} follows from the explicit calculation of Gaussian integrals in \eqref{34-1}. The second equality is a consequence of the representation 
\begin{equation}\label{34-7}
1-is/\beta_j=\sqrt{1+s^2/\beta_j^2}e^{-i\arctan(s/\beta_j)}
\end{equation}
for a complex number in the form of the product of the absolute value by the phase factor.
\end{proof}
\begin{deff}\label{34-d-1}
By symbols $\Phi(s,\beta)$, $f(s)$, and $g(s)$ we notate the limit values of the functions $\Phi_n(s,\beta)$, $f_n(s)$, and $g_n(s)$, when $n\to+\infty$, if such exist.
\end{deff}
\begin{lemma}\label{34-l-2} Let $s\in\mathbb{R}\setminus\{0\}$ and $\beta\in\mathcal{B}_2$. Then, taking into account all of the above, the following statements are true:
	\begin{enumerate}
		\item $0<f(s)<f_{n+1}(s)< f_{n}(s)<1$ for all $n\in\mathbb{N}$;
		\item $f(\cdot)\in L^p(\mathbb{R})$ for all $p>0$;
		\item $g(s)<+\infty$ if and only if $b_1(\beta)<+\infty$;
		\item $\Phi_n(s,\beta)$ has a limit when $n\to+\infty$ if and only if $b_1(\beta)<+\infty$.
	\end{enumerate}
If $s=0$, then we have the relations $\Phi(0,\beta)=1$, $f(0)=1$, and $g(0)=0$.
\end{lemma}
\begin{proof} Note that the number $1+s^2/\beta_j^2\geqslant1$, at the same time, equality is achieved only at the point $s=0$. Therefore, for any $n\in\mathbb{N}$ and $s\in\mathbb{R}\setminus\{0\}$ the estimate holds
\begin{equation}\label{34-11}
f_{n+1}(s)=\prod_{j=1}^{n+1}\Big(1+s^2/\beta_j^2\Big)^{-\frac{1}{4}}<
\prod_{j=1}^{n}\Big(1+s^2/\beta_j^2\Big)^{-\frac{1}{4}}=f_{n}(s)<\ldots<1,
\end{equation}	
which is thus also true for the limiting case.
Next, we note that from the relation $\ln(1+u)\leqslant u$, which is true for all $u\geqslant0$, follows the estimate from the bottom of the form
\begin{equation}\label{34-12}
f(s)=\prod_{j=1}^{+\infty}\exp\bigg(-\frac{1}{4}\ln\big(1+s^2/\beta_j^2\big)\bigg)\geqslant
\exp\bigg(-\frac{s^2}{4}\sum_{j=1}^{+\infty}\beta_j^{-2}\bigg)=e^{-s^2b_2(\beta)/4}>0.
\end{equation}
The second property follows from the fact that the function $f_n(s)$ at $s\to\pm\infty$ has the behavior $|s|^{-n/2}$. Thus, for a fixed $p>0$, all functions $f_n(\cdot)$ with index $n>2/p$ belong to $L^p(\mathbb{R})$. Therefore, given the first property, we obtain $f(\cdot)\in L^p(\mathbb{R})$.

Next, consider $v\in\mathbb{R}$, then the chain of equalities is true
\begin{equation}\label{34-8}
\arctan(v)-v=v\int_0^1\frac{\mathrm{d}t}{1+v^2t^2}-v=-v^3\int_0^1\frac{\mathrm{d}t\,t^2}{1+v^2t^2},
\end{equation}
from which follows the estimate of the form
\begin{equation}\label{34-9}
\Big|\arctan(v)-v\Big|\leqslant|v|^3/3.
\end{equation}
Therefore, substituting in the last relation $v=s/\beta_j$ and summing by the index $j$ from 1 to $n$, we obtain the inequality
\begin{equation}\label{34-10}
\bigg|g_n(s)-s\sum_{j=1}^n\beta_j^{-1}\bigg|\leqslant\frac{s^3}{3}\sum_{j=1}^n\beta_j^{-3}
\leqslant\frac{|s|^3b_2(\beta)}{3\mu(\beta)}<+\infty,
\end{equation}
from which the third part of the statement is derived. The fourth statement is a consequence of the third and formula \eqref{34-5}. The last part follows from the direct substitution of $s=0$ in \eqref{34-6}.
\end{proof}
\begin{deff}\label{34-d-3}
Let $\beta\in\mathcal{B}_2$.
A regularizing function is called a function $\rho(\cdot)\in C\big(\mathbb{R}_+,[0,1]\big)$, such that $\rho(0)=1$, and also the property is fulfilled
\begin{equation}\label{34-13}
\sum_{j=1}^{+\infty}\frac{1}{\beta_j(\Lambda)}=r(\Lambda)+\kappa+o(1)<+\infty,
\,\,\,
\mbox{where}\,\,\,\beta_j(\Lambda)=\beta_j\big/\rho\big(\sqrt{\beta/\Lambda}\big),
\end{equation}
for all $\Lambda$ larger than some fixed value. Here is the function $r(\Lambda)$ denotes the singular component with respect to the parameter $\Lambda$ and tends to $+\infty$ for $\Lambda\to+\infty$. The function $\kappa$ is a constant and does not depend on $\Lambda$.
\end{deff}
\begin{exm}\label{34-e-8}
Under the conditions under consideration, the regularizing function can always be found. As an option, we can suggest a characteristic function of the interval $[0,a]$ for some fixed $a>0$. The choice is not unique.
\end{exm}
\begin{remark}\label{34-r-2}
The functions $r(\cdot)$ and $\kappa$ depend on the spectrum and the regularizing function, but these designations are usually omitted for convenience. Also note that for any $\beta\in\mathcal{B}_1$, the equality $r(\Lambda)=0$ holds. 
\end{remark}
\begin{deff}\label{34-d-4}
Let $\beta\in\mathcal{B}_2\setminus\mathcal{B}_1$, $s\in\mathbb{R}$, and $\rho(\cdot)$ is a regularizing function from Definition \ref{34-d-3}. A regularization of the function $\Phi(s,\beta)$ is the transition to the function $\Phi(s,\beta(\Lambda))$ by the following deformation $$\beta=\{\beta_j\}_{j=1}^{+\infty}\to\beta(\Lambda)=\{\beta_j(\Lambda)\}_{j=1}^{+\infty}.$$
The removal of the regularization denotes the limit transition $\Lambda\to+\infty$.
\end{deff}
\begin{deff}\label{34-d-5}
The object $\Phi(s,\beta(\Lambda))$ from Definition \ref{34-d-4} is called renormalizable if by shifting the internal parameters, $s$ and $\beta(\Lambda)$, as well as changing the common multiplier, absolute value or phase, it is possible to obtain a finite limit after removing the regularization. Formally, the general form of the renormalized function has the form
\begin{equation*}
e^{h_1+ih_2}
\Phi\big(s+h_3,\{\beta_j(\Lambda)+h_{3+j}\}_{j=1}^{+\infty}\big),
\end{equation*}
where $h_i=h_i(\Lambda,s,\beta)$ are real for all $i\in\mathbb{N}$. In this case, all nonzero introduced functions should diverge when removing the regularization for all fixed values of $s$ and should tend to zero at $s\to0$ for each fixed value of $\Lambda$.
\end{deff}
\begin{remark}\label{34-r-12}
It follows from Definition \ref{34-d-5} that if the function $\Phi(s,\beta)$ exists and is finite, then the renormalization does not change it.
\end{remark}
\begin{remark}\label{34-r-13}
The renormalization procedure is not unique. If there are divergences, it leads to a family of functions.
\end{remark}
\begin{remark}\label{34-r-5}
Renormalizability in the general case, when considering other functions, may depend on the introduced regularization, therefore, the regularized function $\Phi(s,\beta(\Lambda))$, and not $\Phi(s,\beta)$, is used in Definition \ref{34-d-5}.
\end{remark}
\begin{theorem}\label{34-t-1}
Let $s\in\mathbb{R}$, $\beta\in\mathcal{B}_2\setminus\mathcal{B}_1$, and also the assumptions from Definitions \ref{34-d-3} and \ref{34-d-4} are fulfilled. The function $\Phi(s,\beta(\Lambda))$ is renormalizable. The renormalization process consists of changing the phase multiplier by the transition
\begin{equation}\label{34-14}
\Phi(s,\beta(\Lambda))\to \Phi(s,\beta(\Lambda))e^{-is(r(\Lambda)+\theta)/2},
\end{equation}
where the function $r(\Lambda)$ is from \eqref{34-13}, and $\theta\in\mathbb{R}$ is a free parameter. The process of removing the regularization after the renormalization leads to the function
\begin{equation}\label{34-16}
	\Phi(s,\beta,\theta)=\Phi(s,\beta(\Lambda))e^{-is(r(\Lambda)+\theta)/2}\bigg|_{\Lambda\to+\infty}=
	f(s)\exp\bigg(-\frac{i}{2}\big(s\theta+g_{\mathrm{r}}(s)\big)\bigg),
\end{equation}
where
\begin{equation}\label{34-19}
g_{\mathrm{r}}(s)=-s\kappa+
s^3\int_0^1\mathrm{d}t\,
\sum_{j=1}^{+\infty}\frac{t^2\beta_j^{-3}}{1+s^2t^2/\beta_j^2}.
\end{equation}
\end{theorem}
\begin{proof} It is necessary to use the representation from Lemma \ref{34-l-1}, the statements from Lemma \ref{34-l-2} and relation \eqref{34-8}. Indeed, in this case, it can be seen that the exponent of the phase multiplier subtracts the singular part of the series, which diverges when the regularization is removed. The remaining part leads to the stated answer.
\end{proof}
\begin{remark}\label{34-r-4}
In a brief form, the manipulations performed above can be represented as the following scheme
\begin{equation}\label{34-15}
	\begin{tikzcd}[column sep=3pc]
		\Phi(s,\beta) \arrow{r}{\footnotesize\mbox{reg.}} & 
		\Phi(s,\beta(\Lambda)) \arrow{r}{\footnotesize\mbox{ren.}} &
		\Phi(s,\beta(\Lambda))e^{-is(r(\Lambda)+\theta)/2} \arrow{r}{\footnotesize\Lambda\to+\infty} &
		\Phi(s,\beta,\theta).
	\end{tikzcd}
\end{equation}
Indeed, starting with the object that was formally the limit of the divergent sequence \eqref{34-1}, that is, it did not actually exist, we got a finite meaningful function. At the same time, such a transition was made in three stages.
\begin{enumerate}
	\item Introduction of the regularization. After that, the object became finite and dependent on the auxiliary regularizing parameter $\Lambda$.
	\item Renormalization. After that, the singular components were removed from the regularized object. At the same time, there was a dependence on the additional parameter $\theta$, which in a sense symbolizes the initial conditions. In applications, this parameter is fixed based on physical and computational considerations.
	\item Removing the regularization. After that, the dependence on the auxiliary regularizing parameter $\Lambda$ disappeared. At the same time, the dependence on the additional parameter $\theta$ remained.
\end{enumerate}
\end{remark}

\begin{deff}\label{34-d-8}
Let $\lambda\geqslant0$ and the assumptions from Definition \ref{34-d-2} are fulfilled. Let us introduce the function
\begin{equation}\label{34-17}
	Z_n(\lambda,\beta)=
	\Bigg(\prod_{j=1}^{n}\int_{\mathbb{R}}\frac{\mathrm{d}a_j}{\sqrt{\pi/\beta_j}}\Bigg)
	\exp\bigg(-S_0[p_n(a)]-\lambda\Big(S_1[p_n(a)]\Big)^2\bigg).
\end{equation}
With the symbol $Z(\lambda,\beta)$ we notate the limit value at $n\to+\infty$, if it exists.
\end{deff}
\begin{lemma}\label{34-l-4}
Let $\beta\in\mathcal{B}_2\setminus\mathcal{B}_1$, then the function $Z(\lambda,\beta)=0$ for all $\lambda>0$. At the same time $Z(1,\beta)=1$.
\end{lemma}
\begin{proof} Note that the functions from \eqref{34-1} and \eqref{34-17} are related by the following integral transformation
\begin{equation}\label{34-20}
	Z_n(\lambda,\beta)=\mathrm{T}\big(\Phi_n(\,\cdot\,,\beta)\big)(\lambda)=
	\frac{1}{\sqrt{4\pi\lambda}}
	\int_{\mathbb{R}}\mathrm{d}s\,e^{-s^2/(4\lambda)}
	\Phi_n(s,\beta),
\end{equation}
which follows from the calculation of the standard Gaussian integral. Define a set of numbers and a function
\begin{equation}\label{34-23}
c_n=\sum_{j=1}^n\beta_j^{-1},\,\,\,h_n(s)=\frac{1}{\sqrt{4\pi\lambda}}e^{-s^2/(4\lambda)}
f_n(s)e^{i(g_n(s)-sc_n)/2}.
\end{equation}
Due to exponential decrease and properties from Lemma \ref{34-l-2}, the function $h_n(s)$ and any of its derivatives belong to $L^1(\mathbb{R})$. Besides, the estimates are correct
\begin{equation}\label{34-25}
\bigg|\frac{\mathrm{d}f_n(s)}{\mathrm{d}s}\bigg|\leqslant\frac{|s|b_2(\beta)}{2},\,\,\,
\bigg|\frac{1}{2}\frac{\mathrm{d}(g_n(s)-sc_n)}{\mathrm{d}s}\bigg|\leqslant
\sum_{j=1}^n\frac{s^2}{2\beta_j^3}\leqslant\frac{s^2b_2(\beta)}{2\mu(\beta)},
\end{equation}
\begin{equation}\label{34-24}
\bigg|\frac{\mathrm{d}h_n(s)}{\mathrm{d}s}\bigg|\leqslant
\frac{e^{-s^2/(4\lambda)}}{\sqrt{4\pi\lambda}}
\bigg(\frac{|s|}{2\lambda}+\frac{|s|b_2(\beta)}{2}+\frac{s^2b_2(\beta)}{2\mu(\beta)}\bigg).
\end{equation}
Note that the right-hand sides do not depend on the parameter $n$, so the inequalities are also valid for the limiting case.
Next, substituting the representation from \eqref{34-5} into \eqref{34-20} and integrating it by parts, we get
\begin{equation}\label{34-26}
Z_n(\lambda,\beta)=\int_{\mathbb{R}}\mathrm{d}s\,e^{isc_n/2}h_n(s)=-
\frac{2}{ic_n}\int_{\mathbb{R}}\mathrm{d}s\,e^{isc_n/2}\frac{\mathrm{d}h_n(s)}{\mathrm{d}s}.
\end{equation}
Therefore, passing to the absolute values of the functions, we obtain an estimate of the form
\begin{equation}\label{34-18}
|Z_n(\lambda,\beta)|\leqslant
\frac{2}{c_n}
\Bigg(
\int_{\mathbb{R}}\mathrm{d}s\,\frac{e^{-s^2/(4\lambda)}}{\sqrt{4\pi\lambda}}
\bigg(\frac{|s|}{2\lambda}+\frac{|s|b_2(\beta)}{2}+\frac{s^2b_2(\beta)}{2\mu(\beta)}\bigg)\Bigg),
\end{equation}
in which the right-hand side tends to zero at $n\to+\infty$ for all fixed $\lambda>0$. The final equality $Z(1,\beta)=1$ follows from direct substitution.
\end{proof}
\begin{theorem}\label{34-t-2}
Let $\lambda>0$, $a\in\mathbb{R}$, and also $\beta\in\mathcal{B}_2\setminus\mathcal{B}_1$. Let us also assume the deformed element $\beta(\Lambda)$ and the function $r(\Lambda)$ are from Definition \ref{34-d-3}. Then the regularization and renormalization process for $Z(\lambda,\beta)$, taking into account all of the above, can be performed according to the following scheme
\begin{equation}\label{34-32}
	\begin{tikzcd}[column sep=3pc]
		\Phi(s,\beta) \arrow{r}{\footnotesize\mbox{reg.}} & 
		\Phi(s,\beta(\Lambda)) \arrow{r}{\footnotesize\mbox{ren.}} \arrow{d}{\mathrm{T}} &
		\Phi(s,\beta(\Lambda))e^{-is(r(\Lambda)+\theta)/2} \arrow{r}{\footnotesize\Lambda\to+\infty} \arrow{d}{\mathrm{T}} &
		\Phi(s,\beta,\theta)\arrow{d}{\mathrm{T}}\\
		Z(\lambda,\beta) \arrow{r}{\footnotesize\mbox{reg.}}  & 
		Z(\lambda,\beta(\Lambda)) \arrow{r}{\footnotesize\mbox{ren.}} &
		Z(\lambda,\beta_{\mathrm{r}}(\Lambda))N(\Lambda) \arrow{r}{\footnotesize\Lambda\to+\infty} &
		Z(\lambda,\beta,\theta)
	\end{tikzcd},
\end{equation}
where the multiplier is
\begin{equation}\label{34-31}
N(\Lambda) 
=e^{-(r(\Lambda)+\theta)^2/4}\prod_{j=1}^{+\infty}\sqrt{\beta(\Lambda)/\beta_{\mathrm{r}}(\Lambda)},
\end{equation}
the renormalized element $\beta_{\mathrm{r}}(\Lambda)\in V$ has the form $\{\beta_j(\Lambda)-r(\Lambda)-\theta\}_{j=1}^{+\infty}$, and we also used the integral operator $\mathrm{T}$ from the proof of Lemma \ref{34-l-4}, the kernel of which is $\exp\big(-s^2/(4\lambda)\big)/\sqrt{4\pi\lambda}$. The resulting function, after removing the regularization, is nonzero for almost all values of $\theta\in\mathbb{R}$ and is equal to
\begin{equation}\label{34-22}
Z(\lambda,\beta,\theta)=\frac{1}{\sqrt{\pi\lambda}}
\int_{\mathbb{R}}\mathrm{d}s\,e^{-s^2/(4\lambda)}f(s)\cos\big(s\theta/2+g_{\mathrm{r}}(s)/2\big),
\end{equation}
where the functions $f(\cdot)$ and $g_{\mathrm{r}}(\cdot)$ are the same as in the formulas \eqref{34-16} and \eqref{34-19}.
\end{theorem}
\begin{proof} In formula \eqref{34-20} it was shown that the functions $\Phi_n(s,\beta)$ and $Z_n(\lambda,\beta)$ are connected by an integral transformation. At the same time, when going to the limit $n\to+\infty$, the function $\Phi(s,\beta)$ does not exist for all $s\in\mathbb{R}\setminus\{0\}$, see Lemma \ref{34-l-2}, while the function $Z(\lambda,\beta)=0$ for all $\lambda>0$, see Lemma \ref{34-l-4}. Let us regularize both objects according to Definition \ref{34-d-4}, that is, by deforming the element $\beta\to\beta(\Lambda)$. In this case, the limit function $\Phi(s,\beta(\Lambda))$ become finite due to the convergence of the series \eqref{34-13} and the third property from Lemma \ref{34-l-2}. At the same time, for the limit function $Z(\lambda,\beta(\Lambda))$ the representation is correct
\begin{equation}\label{34-27}
Z(\lambda,\beta(\Lambda))=\frac{1}{\sqrt{4\pi\lambda}}
\int_{\mathbb{R}}\mathrm{d}s\,e^{-s^2/(4\lambda)}\Phi(s,\beta(\Lambda)),
\end{equation}
this is again verified by a simple calculation of the Gaussian integral. Note that in the regularized case, we have rearranged the integration from the transformation $\mathrm{T}$ and the transition to the limit $n\to+\infty$. This is possible because, see Theorem 1 in paragraph 3 of chapter 14 of the monograph \cite{34-f-1}, firstly, the integral function tends to the limit
\begin{equation*}
e^{-s^2/(4\lambda)}\Phi_n(s,\beta(\Lambda))\to
e^{-s^2/(4\lambda)}\Phi(s,\beta(\Lambda))
\end{equation*}
uniformly with respect to the variable $s$, see Lemma \ref{34-l-2}, and secondly, the integral itself is uniformly convergent with respect to the index $n$.

Next, check that the function $Z(\lambda,\beta(\Lambda))$ is renormalizable. To do this, it must be shown that the singular components can be removed by a suitable shift, $\lambda$ and/or $\beta(\Lambda)$, and changing the common multiplier. Note that the problem with removing the regularization again lies in the singular behavior of the series from \eqref{34-13}, because this leads to large fluctuations in the phase multiplier in \eqref{34-27} and, as a result, to the zeros of the function $Z(\lambda,\beta)$. In this regard, let us first try to renormalize the function $\Phi(s,\beta(\Lambda))$, and then check that this actually leads to a positive result for $Z(\lambda,\beta(\Lambda))$. Thus, a potential candidate for a renormalized function has the form
\begin{equation}\label{34-28}
\frac{1}{\sqrt{4\pi\lambda}}
\int_{\mathbb{R}}\mathrm{d}s\,e^{-s^2/(4\lambda)}\Phi(s,\beta(\Lambda))e^{-is(r(\Lambda)+\theta)/2}.
\end{equation}
Let us use the representation from \eqref{34-3}, then, taking into account the notation from \eqref{34-2}, the last integral is calculated explicitly and is equal to
\begin{equation}\label{34-29}
\Bigg(\prod_{j=1}^{+\infty}\int_{\mathbb{R}}\frac{\mathrm{d}a_j}{\sqrt{\pi/\beta_j(\Lambda)}}\Bigg)
e^{-S_0[a,\beta(\Lambda)]-\lambda\big(S_1[a]-r(\Lambda)/2-\theta/2\big)^2},
\end{equation}
which corresponds to the shift of the element
\begin{equation}\label{34-30}
\beta(\Lambda)\to\beta_{\mathrm{r}}(\Lambda)=\{\beta_j(\Lambda)-\lambda r(\Lambda)-\lambda \theta\}_{j=1}^{+\infty}
\end{equation}
and multiplication by the common multiplier \eqref{34-31}.
Thus, renormalizability is verified. Next, let us move on to the last column of schema \eqref{34-32}. First, let us go to the limit $\Lambda\to+\infty$. The permutation of the limit transition and the integral in \eqref{34-28} is possible due to the presence of uniform convergence. Thus, formula \eqref{34-22} is obtained using the result of Theorem \ref{34-t-1}. 

Let us show that for a fixed $\lambda>0$, the function from \eqref{34-22} is nonzero for almost all $\theta\in\mathbb{R}$. To do this, note that the function $Z(\lambda,\beta,\,\cdot\,)$ is holomorphic in the domain
$$\mathcal{P}=\{z\in\mathbb{C}:|\mathrm{Im}(z)|\leqslant p\}$$
for any fixed $p>0$. Indeed, for $\theta\in\mathcal{P}$ the Taylor series for the function from \eqref{34-22} has the form
\begin{equation}\label{34-33}
	\sum_{k=0}^{+\infty}\frac{\theta^k}{k!}\frac{1}{\sqrt{4\pi\lambda}}
	\int_{\mathbb{R}}\mathrm{d}s\,e^{-s^2/(4\lambda)}f(s)\exp\big(-ig_{\mathrm{r}}(s)/2\big)\bigg(-\frac{is}{2}\bigg)^k
\end{equation}
and it converges absolutely, because, taking into account Lemma \ref{34-l-2}, it admits the followin estimate
\begin{equation*}
\big|\eqref{34-32}\big|\leqslant
\sum_{k=0}^{+\infty}\frac{|\theta|^k}{k!}\frac{1}{\sqrt{4\pi\lambda}}
\int_{\mathbb{R}}\mathrm{d}s\,e^{-s^2/(4\lambda)}\frac{|s|^k}{2^k}=
\frac{1}{\sqrt{4\pi\lambda}}
\int_{\mathbb{R}}\mathrm{d}s\,e^{-s^2/(4\lambda)+|s\theta|/2}<+\infty.
\end{equation*}
Thus, $Z(\lambda,\beta,\,\cdot\,)$ is holomorphic in $\mathcal{P}$ and cannot have a set with limit point inside the domain. Therefore, on any limited interval of $\mathbb{R}$, the function can have only a finite number of zeros.
\end{proof}
\begin{remark}\label{34-r-6}
Note that during the renormalization process, some of the elements of the sequence $\beta\in\mathcal{B}_2$ shifted to negative values, so that $\beta_{\mathrm{r}}(\Lambda)\in V$. Thus, the quadratic form $S_0[\,\cdot\,,\,\cdot\,]$ has ceased to be strictly positive, which leads to additional exponential growth in the integrand of \eqref{34-17}. The rigorous mathematical calculations outlined above can be reformulated on a qualitative level as follows: the exponential growth that arose after the renormalization "balances" the exponential decrease of the fourth degree.
\end{remark}
\section{Physical approach}
\label{34:sec:p}
When integrating polynomials in the finite-dimensional case, as demonstrated by Example \ref{34-e-1} below, the integral and the sum can be rearranged. This makes it possible to reduce the integration process to exponential differentiation, or, to formulate it in physical language, to reduce integration to the application of Wick's theorem on pairings. However, even for the simplest example, see \eqref{34-f-7}, this possibility arises only with significant restrictions on the parameters. Thus, in the field of acceptable parameters, mathematical and physical approaches are equivalent. In other cases, the physical approach operates only with formal (or asymptotic) series. Note that the example below corresponds to a simple case. In real problems, higher-order exponentials arise, that is, $\exp(isa_j^2)\to\exp(isa_j^k)$, where $k>2$, which either leads to even more significant restrictions on the parameters, or makes comparing approaches impossible if the radius of convergence is zero.
\begin{exm}\label{34-e-1} Let $k\in\mathbb{N}$ and the assumptions of Definition \ref{34-d-2} and Lemma \ref{34-l-3} are fulfilled, then the chain of relations is valid
\begin{multline}\label{34-f-5}
\Bigg(\prod_{j=1}^{n}\int_{\mathbb{R}}\frac{\mathrm{d}a_j}{\sqrt{\pi/\beta_j}}
\,e^{-\beta_j^{\phantom{1}}a_j^2}\Bigg)\Big(S_1[p_n(a)]\Big)^k=
\bigg(\sum_{i=1}^n\partial_{c_i}^2\bigg)^k
\Bigg(\prod_{j=1}^{n}\int_{\mathbb{R}}\frac{\mathrm{d}a_j}{\sqrt{\pi/\beta_j}}
\,e^{-\beta_j^{\phantom{1}}a_j^2+c_ja_j}\Bigg)\Bigg|_{c=0}=\\=
\bigg(\sum_{i=1}^n\partial_{c_i}^2\bigg)^k\exp
\bigg(\frac{1}{4}\sum_{j=1}^nc_j^2/b_j\bigg)\bigg|_{c=0}.
\end{multline}
In particular, if we select $n=1$ and $\beta_1=1$, then the result is equal to
\begin{equation}\label{34-f-6}
\partial_{t}^{2k}e^{t^2/4}\Big|_{t=0}=\frac{(2k)!}{k!2^{2k}}.
\end{equation}
Next, multiplying by $(is)^k/k!$ and applying summation over the index $k$, we get
\begin{equation}\label{34-f-7}
\sum_{k=0}^{+\infty}\frac{(is)^k}{k!}\int_{\mathbb{R}}\frac{\mathrm{d}t}{\sqrt{\pi}}\,e^{-t^2}t^{2k}=
\sum_{k=0}^{+\infty}\frac{(is)^k}{k!}\frac{(2k)!}{k!2^{2k}}.
\end{equation}
The last sum converges for $|s|<1$ and is equal to $(1-is)^{-1/2}$, which reproduces the result from the formula \eqref{34-5}. However, for the remaining values of the parameter $s$, the permutation of the integral and the sum is prohibited.
\end{exm}
\begin{remark}\label{34-r-7} In some cases, within the framework of the physical approach, it is possible to use special summations that allow to restore the function and, thus, move into the framework of the mathematical approach. An example is Borel summation, see \cite{34-ha,34-hb}.
\end{remark}

\begin{deff}\label{34-d-9} Let $a\in V$, and also $\beta\in\mathcal{B}_2\setminus\mathcal{B}_1$ according to Definition \ref{34-d-7}. Let us apply the regularization based on Definition \ref{34-d-3}. Let us introduce three elements of diagrammatic technique.
\begin{enumerate}
	\item The dot $\bullet$ denotes the summation operator for all index values.
	\item The line with one index ${\centering\adjincludegraphics[width = 1 cm, valign=c]{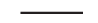}}j$
	denotes the component $a_j$ for $j\in\mathbb{N}$.
	\item The line with two indices $i{\centering\adjincludegraphics[width = 1 cm, valign=c]{34-5.eps}}j$ denotes the value $\delta_{ij}(2\beta_j(\Lambda))^{-1}$, where $i,j\in\mathbb{N}$.
\end{enumerate}
\end{deff}
\begin{exm}\label{34-e-2} Let $a\in V$ be summable with a square, that is, $S_1[a]<+\infty$, then
\begin{equation}\label{34-f-1}
	S_1[a]=\sum_{j=1}^{+\infty}a_j^2=
	{\centering\adjincludegraphics[width = 1.5 cm, valign=c]{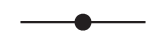}},\,\,\,
	\sum_{i,j=1}^{+\infty}a_i\delta_{ij}(2\beta_j(\Lambda))^{-1}a_j=
	{\centering\adjincludegraphics[width = 2.1 cm, valign=c]{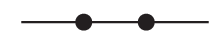}}.
\end{equation}
\end{exm}
\begin{lemma}\label{34-l-5}
Taking into account the assumptions from Definition \ref{34-d-9}, the number $2^{-k}b_k(\beta(\Lambda))$ from Definition \ref{34-d-7}, where $k\in\mathbb{N}$, can be represented as a circle with $k$ dots.
\end{lemma}
\begin{proof} Note that the definition of the number can be rewritten using a multiple summation
\begin{equation}\label{34-f-8}
2^{-k}\sum_{j=1}^{+\infty}\beta_j^{-k}(\Lambda)=
\sum_{j_1,\ldots,j_k=1}^{+\infty}
\Big(\delta_{j_1j_2}\big(2\beta_{j_1}(\Lambda)\big)^{-1}\Big)\cdot\ldots\cdot
\Big(\delta_{j_kj_1}\big(2\beta_{j_k}(\Lambda)\big)^{-1}\Big),
\end{equation}
then the statement follows from the application of the 1st and 3rd elements of the diagrammatic technique.
\end{proof}
\begin{exm}\label{34-e-3}
Taking into account the assumptions from Definition \ref{34-d-9} and the result of Lemma \ref{34-l-5}, the equalities are true
\begin{equation}\label{34-f-9}
b_1(\beta(\Lambda))=2
{\centering\adjincludegraphics[width = 0.9  cm, valign=c]{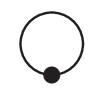}},\,\,\,
b_2(\beta(\Lambda))=4
{\centering\adjincludegraphics[width = 1.1 cm, valign=c]{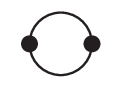}},\,\,\,
b_3(\beta(\Lambda))=8
{\centering\adjincludegraphics[width = 1 cm, valign=c]{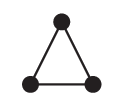}}.
\end{equation}
\end{exm}

\begin{deff}\label{34-d-10}
Let the assumptions from Definition \ref{34-d-9} be true. The diagrammatic representation \eqref{34-f-1} for $S_1[a]$ is called a vertex and denoted by the same symbol $S_1$ without an argument. The symbol $\mathbb{H}$ denotes a linear operator from the set of polynomials $\mathbb{C}[S_1]$ to $\mathbb{C}$, which is defined on monomials $S_1^k$ by two rules.
\begin{enumerate}
	\item If $k=0$, then $\mathbb{H}(S_1^0)=\mathbb{H}(1)=1$.
	\item If $k\in\mathbb{N}$, then $\mathbb{H}(S_1^k)$ is equal to the sum of all possible pairwise connections of $2k$ pieces of external lines for $k$ vertices $S_1$.
\end{enumerate}
\end{deff}
\begin{exm}\label{34-e-4}
Taking into account Definitions \ref{34-d-9} and \ref{34-d-10}, the following relations hold
\end{exm}
\begin{equation}\label{34-f-2}
\mathbb{H}(S_1^{\phantom{1}})=
{\centering\adjincludegraphics[width = 0.9 cm, valign=c]{34-2.eps}}=
\frac{1}{2}b_1(\beta(\Lambda))
,
\end{equation}
\begin{equation}\label{34-f-4}
\mathbb{H}(S_1^2)=
{\centering\adjincludegraphics[width = 0.9  cm, valign=c]{34-2.eps}}
{\centering\adjincludegraphics[width = 0.9  cm, valign=c]{34-2.eps}}
+2
{\centering\adjincludegraphics[width = 1.1 cm, valign=c]{34-3.eps}}=
\frac{1}{4}b_1^2(\beta(\Lambda))+
\frac{1}{2}b_2(\beta(\Lambda))
,
\end{equation}
\begin{equation}\label{34-f-3}
\mathbb{H}(S_1^3)=
{\centering\adjincludegraphics[width = 0.9  cm, valign=c]{34-2.eps}}
{\centering\adjincludegraphics[width = 0.9  cm, valign=c]{34-2.eps}}
{\centering\adjincludegraphics[width = 0.9  cm, valign=c]{34-2.eps}}+6
{\centering\adjincludegraphics[width = 0.9  cm, valign=c]{34-2.eps}}
{\centering\adjincludegraphics[width = 1.1 cm, valign=c]{34-3.eps}}+8
{\centering\adjincludegraphics[width = 1 cm, valign=c]{34-4.eps}}=
\frac{1}{8}b_1^3(\beta(\Lambda))+
\frac{3}{4}b_1(\beta(\Lambda))b_2(\beta(\Lambda))+
b_3(\beta(\Lambda)).
\end{equation}
\begin{deff}\label{34-d-11}
Let the assumptions from Definition \ref{34-d-10} be correct. Using the symbol $\mathbb{H}_1$, we notate a linear operator from the set of polynomials $\mathbb{C}[S_1]$ to $\mathbb{C}$, which is given by the formula
\begin{equation}\label{34-f-10}
\mathbb{H}_1(S_1^k)=\mathbb{H}(S_1^k)\Big|_{b_1\equiv0}
\end{equation}
for all $k\in\mathbb{N}\cup\{0\}$.
\end{deff}
\begin{exm}\label{34-e-5}
Taking into account Definitions \ref{34-d-9} and \ref{34-d-11}, the following relations hold
\begin{equation}\label{34-f-11}
\mathbb{H}_1(1)=1,\,\,\,
\mathbb{H}_1(S_1^{\phantom{1}})=0,\,\,\,
\mathbb{H}_1(S_1^2)=2
{\centering\adjincludegraphics[width = 1.1 cm, valign=c]{34-3.eps}}=
\frac{1}{2}b_2(\beta(\Lambda)),\,\,\,
	\mathbb{H}_1(S_1^3)=8
{\centering\adjincludegraphics[width = 1 cm, valign=c]{34-4.eps}}=
b_3(\beta(\Lambda)).
\end{equation}
\end{exm}
\begin{deff}\label{34-d-12}
Let the assumptions from the definitions outlined above be true. Let also $f(\cdot)$ be a function that can be represented as a Taylor series 
\begin{equation}\label{34-f-12}
f(z)=\sum_{j=0}^{+\infty}f_jz^j
\end{equation}
for all argument values $z\in\mathbb{C}$. Then the regularized "physical" functional integral of the function $f(S_1[\,\cdot\,])$ with a weight $\exp(-S_0[\,\cdot\,,\beta(\Lambda)])$ is defined as a formal series
\begin{equation}\label{34-f-13}
\sum_{j=0}^{+\infty}f_j\mathbb{H}(S_1^j),
\end{equation}
which is indicated by a symbolic entry of the form
\begin{equation}\label{34-f-14}
\int_V\mathcal{D}a\,e^{-S_0[a,\beta(\Lambda)]}f(S_1[a]).
\end{equation}
\end{deff}
\begin{exm}\label{34-e-6}
Taking into account Definition \ref{34-d-12}, the models studied in Section \ref{34:sec:m} are formulated as follows
\begin{equation}\label{34-f-15}
\int_V\mathcal{D}a\,e^{-S_0[a,\beta(\Lambda)]+isS_1[a]}=
\sum_{j=0}^{+\infty}\frac{(is)^j}{j!}\mathbb{H}(S_1^j),
\end{equation}
\begin{equation}\label{34-f-16}
\int_V\mathcal{D}a\,e^{-S_0[a,\beta(\Lambda)]-\lambda S_1^2[a]}=
\sum_{j=0}^{+\infty}\frac{(-\lambda)^j}{j!}\mathbb{H}(S_1^{2j}).
\end{equation}
\end{exm}
\begin{remark}\label{34-r-8}
Note that the physical approach is much narrower. Indeed, in Section \ref{34:sec:m}, both models were connected by an integral transformation. In this case, such a relation is preserved only if the integral transformation applied to the formal series is understood as the transformation applied to each individual coefficient. This procedure leads to a new formal series.
\end{remark}
\begin{remark}\label{34-r-9}
Note that the divergences in the two approaches have different forms and meanings. The following table shows the problems that arise when removing the regularization, that is, in the limit $\Lambda\to+\infty$.
\begin{center}
\renewcommand{\arraystretch}{2}
\begin{tabular}{|c|c|}
	\hline
	Mathematical approach & Physical approach\\
	\hline
	 The phase of $\Phi(s,\beta(\Lambda))$ oscillates rapidly&\multirow{2}{*}{The terms with $b_1(\beta(\Lambda))$ diverge in \eqref{34-f-15} and \eqref{34-f-16}} \\
	\cline{1-1}
	 The function $Z(\lambda,\beta(\Lambda))$ tends to zero &\\
	\hline
\end{tabular}
\renewcommand{\arraystretch}{1}
\end{center}
Thus, the mathematical approach is more detailed. From a physical point of view, the divergences are the same in the two models, and the coefficients of the series become infinitely large. In turn, there is a clear distinction within Section \ref{34:sec:m}. In one case, when the regularization is removed, the phase oscillates uncontrollably, while in the other case, the function identically becomes zero.
\end{remark}
\begin{remark}\label{34-r-10}
When working with formal series, the numbers $\beta_j^{-1}(\Lambda)$ occur in the definition of an element of diagrammatic technique, which is why any shift of the form $\beta_j(\Lambda)\to\beta_j(\Lambda)+\alpha_j$ leads to the need to do one more expansion of the formal series. To circumvent this process, it suffices to note that the specified shift is equivalent to the appearance of the term 
$$-\sum_{j=1}^{+\infty}a_j^2\alpha_j^{\phantom{1}}$$ 
in the exponential function. In the special case, if $a_j=\alpha\in\mathbb{R}$ for all $j\in\mathbb{N}$, the shift is $-\alpha S_1[a]$. This addition leads to the appearance of an additional vertex.
\end{remark}
\begin{deff}\label{34-d-13}
A formal series \eqref{34-f-15} is called renormalizable if it is renormalizable in the sense of Definition \ref{34-d-5}, taking into account two changes.
\begin{enumerate}
	\item The shift $\beta(\Lambda)$ is understood taking into account Remark \ref{34-r-10}.
	\item When removing the regularization, the final value should not be the series as a whole, but only each individual coefficient.
\end{enumerate}
The concept of renormalizability is extended to the series from \eqref{34-f-16} by replacing the parameters $s\leftrightarrow\lambda$ in the formulation.
\end{deff}
\begin{theorem}\label{34-t-3}
Let $\beta\in\mathcal{B}_2\setminus\mathcal{B}_1$ and the assumptions of Definition \ref{34-d-12} are fulfilled, then the formal series \eqref{34-f-15} and \eqref{34-f-16} are renormalizable, and the renormalization process itself consists in passing to formal series of the form
\begin{equation}\label{34-f-17}
	\int_V\mathcal{D}a\,e^{-S_0[a,\beta(\Lambda)]+is(S_1[a]-r(\Lambda)/2-\theta/2)}=
	\sum_{j=0}^{+\infty}\frac{(is)^j}{j!}\mathbb{H}\big((S_1-r(\Lambda)/2-\theta/2)^j\big),
\end{equation}
\begin{equation}\label{34-f-18}
	\int_V\mathcal{D}a\,e^{-S_0[a,\beta(\Lambda)]-\lambda (S_1[a]-r(\Lambda)/2-\theta/2)^2}=
	\sum_{j=0}^{+\infty}\frac{(-\lambda)^j}{j!}\mathbb{H}((S_1-r(\Lambda)/2-\theta/2)^{2j}),
\end{equation}
where $\theta\in\mathbb{R}$ and the function $r(\Lambda)$ is selected according to Definition \ref{34-d-3}.
\end{theorem}
\begin{proof} For convenience, we notate $\zeta=r(\Lambda)/2+\theta/2$. Next, note that $\mathbb{H}(S_1^j)$ can be decomposed into a sum by powers $\xi=b_1(\beta(\Lambda))/2$ in the following form
\begin{equation}\label{34-f-19}
\mathbb{H}(S_1^j)=\sum_{i=0}^jC_i^j\mathbb{H}_1(S_1^i)\xi^{j-i},
\end{equation}
where the coefficient is equal to the number of ways to select $i$ elements from $j$ pieces. Then we perform the chain of transformations
\begin{align*}
\mathbb{H}\big((S_1-\zeta)^n\big)=&\sum_{j=0}^nC_j^n\mathbb{H}(S_1^j)(-\zeta)^{n-j}=
\sum_{j=0}^n\sum_{i=0}^jC_j^nC_i^j\mathbb{H}_1(S_1^i)\xi^{j-i}(-\zeta)^{n-j}\\=&
\sum_{i=0}^n\mathbb{H}_1(S_1^i)
\sum_{j=i}^nC_j^nC_i^j\xi^{j-i}(-\zeta)^{n-j}\\=&
\sum_{i=0}^n\mathbb{H}_1(S_1^i)
\sum_{j=0}^{n-i}C_{j+i}^nC_i^{j+i}\xi^{j}(-\zeta)^{n-j-i}\\=&
\sum_{i=0}^n\mathbb{H}_1(S_1^i)C_{n-i}^n(\xi-\zeta)^{n-i}=\mathbb{H}_1\big((S_1+\xi-\zeta)^n\big),
\end{align*}
from which it follows that all coefficients of the formal series \eqref{34-f-17} and \eqref{34-f-18} are finite because $\beta\in\mathcal{B}_2\setminus\mathcal{B}_1$ and the combination $\xi-\zeta$  is finite before and after removing the regularization.
\end{proof}

\section{Conclusion}
\label{34:sec:zac}

The paper we considered the regularization and the renormalization for a special model both within the framework of working with formal series (the perturbative approach) and using the explicit construction of the functional integral in the sense of a transition to an infinitely large number of dimensions (the limit transition in the finite-dimensional case). It has been clearly demonstrated that the divergences that appear in the first case are of the same nature, while in the second case they can be expressed either through phase oscillation or through zeroing of the entire functional.

One of the interesting and important generalization options for the considered model \eqref{34-17} is the transition
\begin{equation*}
S_1^2[a]=
\sum_{i_1,i_2,i_3,i_4=1}^{+\infty}\delta_{i_1i_2}\delta_{i_3i_4}a_{i_1}a_{i_2}a_{i_3}a_{i_4}\to
\sum_{i_1,i_2,i_3,i_4=1}^{+\infty}C_{i_1i_2i_3i_4}a_{i_1}a_{i_2}a_{i_3}a_{i_4},
\end{equation*}
where the coefficients $C_{i_1i_2i_3i_4}$ have a more complex structure. Using the notation from Example \ref{34-e-7}, we can select coefficients in the form
\begin{equation*}
C_{i_1i_2i_3i_4}=\int_{\mathbf{B}}\mathrm{d}^nx\,\phi_{i_1}(x)\phi_{i_2}(x)\phi_{i_3}(x)\phi_{i_4}(x),
\end{equation*}
which lead to a model with quartic interaction. It is known that for $n\leqslant4$, such a model is renormalizable within the framework of the perturbative approach, see \cite{29-3,29-4,Iv-2024-1,Kh-2024}. 

Note that the consideration of the finite-dimensional case in Section \ref{34:sec:m} can also be considered as some type of regularization. In applications, this approach occurs less frequently, since working with a deformed Green's function (with a deformed spectrum) is more attractive from the point of view of numerical analysis of diagrams.

Note that the renormalization process formulated in Definition \ref{34-d-5} does not make the introduced functions $h_i$ unique. Moreover, it is not clear whether different combinations of singular parts can be used. This issue is open and requires additional research. Intuitively, it can be assumed that the renormalization process in the physical approach, which focuses on each coefficient individually, contains additional constraints that make the process more specific.

\vspace{2mm}
\noindent\textbf{Acknowledgements.} The author is grateful to N.V.Kharuk for useful comments.

\end{document}